\documentclass[journal]{article}                                  

\usepackage{arxiv}
\usepackage{hyperref}
\usepackage{graphicx}          
\usepackage{amsmath} 
\usepackage{amssymb}  

\usepackage{amsthm}
\usepackage{color,soul}
\usepackage{flushend}
\usepackage{cite}
\usepackage{mfirstuc}
\usepackage{subfig}
\usepackage{mathtools}
\usepackage{xcolor}
\usepackage{array,multirow}
\usepackage{algorithm} 
\usepackage{algorithmic}  

\usepackage[linesnumbered,algo2e,ruled,vlined,norelsize]{algorithm2e} 
\usepackage[bb=libus]{mathalpha}
\usepackage{bbm}

\allowdisplaybreaks
\newtheorem{thmm}{Theorem}

\newtheorem{remm}{Remark}

\newtheorem{Lemm}{Lemma}

\usepackage{tikz}
\usepackage{textcomp}
\usepackage{hyperref}
\usepackage{lipsum}

\newcommand\copyrighttext{%
  \footnotesize \textcopyright 2025 IEEE. This work has been accepted to the 2026 American Control Conference (ACC). Personal use of this material is permitted.
  Permission from IEEE must be obtained for all other uses, in any current or future 
  media, including reprinting/republishing this material for advertising or promotional 
  purposes, creating new collective works, for resale or redistribution to servers or 
  lists, or reuse of any copyrighted component of this work in other works. 
  DOI: \href{<http://tex.stackexchange.com>}{To be generated.}}
\newcommand\copyrightnotice{%
\begin{tikzpicture}[remember picture,overlay]
\node[anchor=south,yshift=10pt] at (current page.south) {\fbox{\parbox{\dimexpr\textwidth-\fboxsep-\fboxrule\relax}{\copyrighttext}}};
\end{tikzpicture}%
}

\usepackage{authblk}
\title{\LARGE \bf
KAN-Koopman Based Rapid Detection Of Battery Thermal Anomalies With Diagnostics Guarantees
}

\author[1]{Sanchita Ghosh}
\author[1]{Tanushree Roy}
\affil[1]{Department of  Mechanical Engineering, Texas Tech University, Lubbock, TX 79409, US. Emails:~{\tt\small sancghos@ttu.edu, tanushree.roy@ttu.edu}.}

\begin{document}

\copyrightnotice

\maketitle
\thispagestyle{empty}
\pagestyle{empty}


\begin{abstract}
 Early diagnosis of battery thermal anomalies is crucial to ensure safe and reliable battery operation by preventing catastrophic thermal failures. Battery diagnostics primarily rely on battery surface temperature measurements and/or estimation of core temperatures. However, aging-induced changes in the battery model and limited training data remain major challenges for model-based and machine-learning based battery state estimation and diagnostics. To address these issues, we propose a Kolomogorov-Arnold network (KAN) in conjunction with a Koopman-based detection algorithm that leverages the unique advantages of both methods. Firstly, the lightweight KAN provides a model-free estimation of the core temperature to ensure rapid detection of battery thermal anomalies.  Secondly, the Koopman operator is learned in real time using the estimated core temperature from KAN and the measured surface temperature of the battery to provide the core and surface temperature prediction for diagnostic residual generation. This online learning approach overcomes the challenges of model changes.
 Furthermore, we derive analytical conditions to obtain diagnostic guarantees on our KAN-Koopman detection scheme. Our simulation results illustrate a significant reduction in detection time with the proposed algorithm compared to the baseline Koopman-only algorithm.

\end{abstract}

\section{Introduction} 
Thermal safety of batteries remains a major concern with increasing deployment of batteries in diverse applications ranging from consumer electronics to electric vehicles \cite{klink2021thermal}. Thermal anomalies in the battery can  lead to overheating, accelerated degradation, or even uncontrollable thermal runaway, resulting in fire, venting, or electrolyte leakage \cite{zhou2023chebyshev,wei2019lyapunov}. Therefore, early detection of thermal anomalies is crucial to prevent catastrophic failures \cite{dey2017model}.
Consequently, \cite{dey2017model} proposed a 1D distributed parameter model-integrated observer for thermal fault detection and estimation in a cylindrical battery. 
Similarly, \cite{son2019model} proposed a two-state thermal model-based algorithm to identify thermal faults. 
Furthermore, \cite{sattarzadeh2021thermal}  optimized sensor location to improve thermal fault detectability and isolability  in a pouch cell and proposed a bank of diagnostic filters for each sensor zone. \cite{zhou2023chebyshev} adopted a Chebyshev-Galerkin-based model to construct the spatial basis function with time coefficients to detect and localize thermal faults. \cite{klink2021thermal} utilizes the temperature-dependent deviation of cell-impedance for early thermal fault detection. Moreover, \cite{wei2019lyapunov}, proposed a thermal fault detector that uses battery internal resistance estimation and surface temperature measurements. A  set-based fault detection using constrained zonotopes was proposed in \cite{saccani2022model}. However, the effectiveness of model-based methods heavily depends on parameter reliability, and such methods exhibit a high computational complexity \cite{richardson2014battery}.

To overcome this limitation, \cite{ojo2020neural} proposed a long short-term memory (LSTM)-based neural network (NN) model that utilizes surface temperature for thermal fault detection. Moreover, \cite{sun2022online} adopted the discrete Fréchet distance and local outlier factor to detect and identify faulty battery cells. Additionally, physics-informed NN model \cite{naguib2025thermal} and Principal Component Analysis with mean-based residual generation \cite{bhaskar2023data} have been adopted to detect battery thermal anomalies.  Nevertheless, the data-driven approaches require a large amount of data and frequently become inefficient for unforeseen anomalies \cite{sattarzadeh2021thermal}. Additionally, both model-based and data-driven methods exhibit poor adaptability with varied operating conditions and battery aging \cite{GHOSH2025AE}. 

Authors proposed a model-free Koopman operator-based detection of internal short circuit- \cite{ghosh2025koopman} and corrupted charging- \cite{GHOSH2025AE} induced thermal anomalies in the battery. In addition, the proposed algorithms adopt an online-learning framework that utilizes the limited available data and exhibits an improved adaptability with changes in battery dynamics. However, the methods rely on battery voltage and surface temperature measurements, which can lead to delayed detection of thermal anomalies \cite{karnehm2025core}. Furthermore, while incorporating core temperature information can significantly aid in early detection of thermal anomalies, online core temperature measurement, even with embedded micro-sensors, remains impractical during real-world applications due to its intrusive nature and manufacturing limitations \cite{richardson2014battery,sattarzadeh2021real}. To address these research gaps,  main contributions of this paper are:
\begin{enumerate}
    \item We propose a model-free Kolmogorov-Arnold network (KAN) integrated Koopman-based diagnostic algorithm for the detection of battery thermal anomalies.
        \item We incorporate KAN-based core temperature estimation to ensure faster detection of thermal anomalies.
    \item  We derive analytical conditions to obtain  diagnostic guarantees for our proposed method. 
\end{enumerate}

The rest of the paper is organized as follows. Section~\ref{prel} presents a brief overview of the KAN network and the Koopman operator, and Section~\ref{prob stat} describes our problem framework. In Section~\ref{KK_detection}, we introduce our proposed KAN-Koopman diagnostic algorithm.
The simulation results are presented in Section~\ref{sim_result}. Finally, Section~\ref{conclusion} concludes the paper. 
\section{Preliminaries } \label{prel}

\subsection{Kolomogorov-Arnold network}\label{kan_preli}
Kolmogorov-Arnold network (KAN) has emerged as a promising deep learning model that offers enhanced robustness and precision with a lightweight structure \cite{liu2024kan}. A NN utilizes learnable edge weights and fixed nonlinear activation functions on the output nodes for effective learning. In contrast,  KAN employs learnable nonlinear activation functions on the edges that are added at output nodes for improved nonlinearity representation.
Let us consider a $D-$deep KAN network such that the $d^{th}$ KAN layer $\forall d \in \{1, \cdots, D\}$ has $\omega_d$ input nodes (a width of $\omega_d$) and $\omega_{d+1}$ output nodes, and these output nodes become the input nodes for the next layer $d+1$.   The $a^{th}$ input node and the $b^{th}$ output node in layer $d$, respectively, represent  single variables $\alpha_{d,a}$ and $\alpha_{d+1,b}$. Then, $ \forall a \in\{1,\cdots, \omega_d\} $ and $ \forall b \in\{1,\cdots, \omega_{d+1}\}$,   $\alpha_{d,a}$  pass through the activation functions $\varphi_{d,b,a}$ and then they are added to get  the output variables $\alpha_{d+1,b}$, i.e.\, $\alpha_{d+1,b}=\sum_{a=1}^{\omega_d}\varphi_{d,b,a} (\alpha_{d,a})$. 
 In a compact matrix form, we obtain: 
\begin{align}
    & \boldsymbol{\alpha}_{d+1} =  {\phi_{d}} \boldsymbol{\alpha}_{d}, \quad \text{where} \,\,\, \boldsymbol{\alpha}_{(.)}\! = \!\begin{bmatrix}
        \alpha_{(.),1} \!\! & \!\! \cdots  \!& \!\! \alpha_{(.),\omega_{(.)}}
    \end{bmatrix}^T, \label{KAN x} \\
    & \qquad \,\, {\phi_{d}} := \begin{pmatrix}
\varphi_{d,1,1} \!\!& \!\! \dots \!\!& \!\! \varphi_{d,1,\omega_d}\\
\vdots \!\!& \!\! \ddots \!\!& \!\! \vdots\\
\varphi_{d,\omega_{d+1},1} \!\!& \!\! \dots \!\!& \!\! \varphi_{d,\omega_{d+1},\omega_d}
    \end{pmatrix}. \label{KAN matrix}
\end{align}
    $\phi_d$ is the activation function matrix at layer $d$. 
    Following \cite{liu2024kan}, to represent the activation functions in $\phi_d$, we consider a weighted sum of a silu basis function and a spline function that is learn as a linear combination of $\kappa$-order B-splines over $G+1$ grid points ($G$ intervals). 
    During approximating a function with such a KAN network, we can find a bound for error using the approximation theory (Theorem 2.1) presented in \cite{liu2024kan}. We briefly restate the theorem below.

        \noindent
    \begin{thmm}[Approximation Theory \cite{liu2024kan}] \label{KAN aprox thm}
        Let us consider a bounded multi-variate function $f(\boldsymbol{\alpha})$ for $\boldsymbol{\alpha} = (\alpha_1, \cdots, \alpha_n)$ that can be expressed with finite additive compositions of $\kappa-1$ times differentiable univariate functions $\varphi_{d,b,a}$ such that $f(\boldsymbol{\alpha}) = (\phi_D \circ \phi_{D-1}\circ \cdots \circ \phi_1) \boldsymbol{\alpha}$, where the function matrix $\phi_d,\, \forall \in \{1, \cdots, D\}$ is defined in \eqref{KAN matrix}. Let us also consider a $D-$layer KAN network to approximate this function $f(\boldsymbol{\alpha})$. Then, there exists  $\kappa^{th}$ order B-spline functions $\varphi_{d,b,a}^G$ learned over a finite grid size $G$ such that        the bound on approximation error for the KAN network is:
        \begin{align}
            \lvert f - (\phi_D^G \circ \phi_{D-1}^G\circ \cdots \circ \phi_1^G) \rvert_\infty \leqslant \mathcal{D} G^{-\kappa-1}, \label{error_bound}
        \end{align}
        where, the constant $\mathcal{D}$ depends on the function and its representation, and $\lvert \cdot \rvert_\infty$ denotes the $l_\infty$ norm.
 \end{thmm}

We will use the findings from this theorem later for our proposed algorithm.

\subsection{Koopman operator}\label{KO_preli}
 Koopman operator (KO) is a powerful tool for model-free analysis of dynamical systems with uncertainty, nonlinearity, and complexity using computationally efficient algorithms and limited data \cite{budivsic2012applied}. 
Let us define the nonlinear dynamics of the system as $ {x}_{k+1} = g(x_k,u_k)$ with measurements $y_k = h(x_k)$. Here $x  \in  \mathbb{R}^i$, $u  \in \mathbb{R}^j$, and $y \in \mathbb{R}^l$  are, respectively,  state vector, control input,  and measured output. $h : \mathbf{X} \rightarrow \mathbb{R}^l$ denotes the nonlinear output function, and the system dynamics $g: \mathbb{R}^i \times \mathbb{R}^j \rightarrow \mathbb{R}^i$ is a continuously differentiable nonlinear function. 
To define KO to approximate this system, let us consider an infinite-dimensional Hilbert space of observable functions $\mathcal{F}$. Then,  the set of KO on this space, $\mathcal{K} : \mathcal{F} \rightarrow \mathcal{F}$ is defined using the complex-valued Koopman eigenfunctions (KEF), $\theta \in \mathcal{F}$ such that $\theta :\mathbb{R}^i \rightarrow \mathbb{C}$, and the corresponding Koopman eigenvalues (KEV), $\lambda \in \mathbb{C}$, as \cite{budivsic2012applied}:
\begin{align} \label{kt_phi}
    \left[ \mathcal{K} \theta \right] (x_k,u_k) = \theta \left(g(x_k,u_k) \right) = \lambda  \theta (x_k,u_k); 
\end{align}
We assume that the space $\mathcal{F}$ is spanned by infinitely many KEFs  as $\mathcal{F} =\text{span} \{\theta_p \}_{p = 1}^\infty$ and the output function $h(x_k)$ also lies on $\mathcal{F}$.  Then  $h(x)$  can be expanded in terms of KEFs as $
    h(x_k) = \sum_{p=1}^{\infty} \theta_p (x_k,u_k) v_p^h$ \cite{surana2016linear}.
Here the Koopman Modes (KM) $v_p^{h} \in \mathbb{C}^l$ are the projection coefficients of ${h}(x)$ onto the space $\mathcal{F}$. 
Next, we utilize the delay embedding to obtain the finite approximation of the linear model as ${z}_{k+1} = Az_k + Bu_k$; $ y_k = C z_k$ \cite{budivsic2012applied}.
Here, $ z_k$, $A, B$, and $C$ contain, respectively,  the KEFs $\theta_i$s,  KEVs $\lambda_i$s, and  KMs $v^h_i$s.
We deploy delay embedding over a moving learning window of $W_l$ observations to approximate the Koopman model, and  generate predictions for the next $W_p$ observations, such that the generated predictions for $y_k$ and the evolved KEFs ${z}_k$ over $W_p $ remain bounded. The windows slide ahead with $W_p$ amount after each cycle of prediction. In particular, we rearrange the available data during the learning window to obtain a Hankel matrix $X_b$ using the first $W_l-1$ observations, a shifted matrix $X_s$ using the last $W_l-1$ observations, and an input matrix $U_b$ (detailed formulation of these matrices can be found in \cite{GHOSH2025AE}). 
 Then, the Koopman model  can be approximated solving the optimization problems posed as:
 \vspace{-1mm}
\begin{align} \label{ko opt}
    \min\limits_{\Lambda} \Vert X_{s} -  \, \Lambda  \Upsilon \Vert, \qquad 
   \min\limits_C \Vert Y_b -  \, C X_b \Vert.
\end{align} \
Here,  $\Upsilon^T \!\!= \!\begin{bmatrix}
    X_{b} \!&\!
    U_b
\end{bmatrix}$,  $ \Lambda \!=\! \begin{bmatrix}
    A \!&\! B
\end{bmatrix}$.
With these preliminaries on KAN and KO, next, we discuss our problem framework.

\section{Thermal Anomalies in Battery} \label{prob stat}
Thermal anomalies in a battery cause unanticipated heat generation, leading to overheating or even catastrophic thermal runaway failure. In our framework, we consider two classes of anomalies for such abnormal thermal behavior.

\noindent
\textbf{(I.) Physical faults:} Several factors such as short circuits, cell-reactions, electrolyte leakage, and manufacturing malfunctions can result in thermal anomalies in the battery \cite{sattarzadeh2021thermal}.

\noindent
\textbf{(II.) Compromised charging:} An adversary can corrupt the charging command from cloud-BMS, leading to excessive and unwanted ohmic heat generation in the battery \cite{GHOSH2025AE}.

\noindent
Mathematically, we model the battery thermal dynamics under such anomalies as follows.
\begin{align} 
    & T_{1_{k+1}}  = T_{1_{k}} +\Delta t\left[ -\frac{T_{1_{k}}-T_{2_{k}}}{R_1\,C_1} + \frac{\dot{Q}_{k}}{C_1} + \frac{\Bar{\delta}_{1_k}}{C_1}\right], \label{core temp}\\
    & T_{2_{k+1}}  = T_{2_{k}}+ \Delta t \left[-\frac{T_{2_{k}}-T_{1_{k}}}{R_1\,C_2} - \frac{T_{2_{k}}-T_{\infty_{k}}}{R_2\,C_2}\right], \label{surf temp}\\
   & T_{\infty_{k+1}}  = T_{\infty_{k}}+\Delta t\left[-\frac{T_{\infty_{k}}-T_{2_{k}}}{R_2\,C_{\infty}} - \frac{\dot Q_{c_{k}}}{C_{\infty}}\right], \label{amb temp} \\
    & \dot{Q}_k = \big(I_k+\delta_{2_k} \big)\big(V_{O_k}(\mathbf{s_k}) -V_{t_k}-T_{1_k}\gamma\big). \label{qdot dynamics}
\end{align}
Here,  $T_1, T_2, T_\infty$ are the temperatures and $C_1, C_2, C_\infty$ are the heat capacities of the battery core, surface, and coolant materials, respectively. $R_1$  and $R_2$ denote the thermal resistances, respectively, between the core \& surface, and the surface \& coolant. $\dot{Q}$ is the internal heat generation, $\dot{Q}_c$ is the coolant power,  and $\gamma$ is the entropic heat coefficient. Moreover, $I$ is the charging current, $\mathbf{s}$ is the state-of-charge (SOC), $V_O$ is the open-circuit voltage, and $V_t$ is the terminal voltage of the battery.  $\Bar{\delta}_1$ and $\delta_2$ denote thermal anomalies, where $\delta_2$ can be induced from either a physical fault or a cyberattack, and only a physical fault can induce $\Bar{\delta}_1$. $\Delta t$ is the sampling time. Battery electrical dynamics is defined as
\vspace{-1mm}
\begin{align} 
   \mathbf{s}_{k+1}  = \mathbf{s}_{k}-\frac{\Delta t \Tilde{I}_{k}}{C_b}, \quad \,\, V_{t_k}  = V_{O_k}(\mathbf{s_k}) - \Tilde{I}_kR_b. \label{vt dynamics}
\end{align}
Here,   $\Tilde{I}_k = I_k+\delta_{2_k}$ is corrupted charging actuation, $C_b$ is the discharge capacity, and $R_b$ is the internal resistance of the battery. Furthermore, substituting $V_t$ in \eqref{qdot dynamics}, we get $\dot{Q}_k =  \Tilde{I}_k^2R_b- \Tilde{I}_k\,T_{1_k} \gamma$. 
We consider that temperature measurements of the surface  $T_2$ and  ambient $T_\infty$, along with the nominal  current input $I$, and coolant input $\dot{Q}_c$ are available to us.

\section{KAN-Koopman Diagnostic Algorithm}\label{KK_detection}
The KAN-Koopman-based detection algorithm has three main steps: (i) Estimation of core temperature using an offline-trained KAN network. (ii) Predictions of core and surface temperature using an online-learned Koopman model. (iii) Generation of diagnostic residual as Koopman prediction errors to detect thermal anomalies. The  algorithm addresses the challenges in battery thermal diagnostics as follows.

\textbf{(I) Model-free approach:}  Our proposed algorithm adopts a data-driven  approach to detect thermal anomalies using only the available data ($T_2, T_\infty, I, \dot{Q}_c$), and without relying on the  prior model knowledge of battery thermal dynamics.

\textbf{(II) Enhanced thermal information with KAN-based $T_1$ estimation:} The algorithm trains a KAN network to effectively learn intrinsic thermal patterns offline, where  $T_1$ measurements can be accessed using a lab experimental set-up or high-fidelity simulations. The trained KAN model is then deployed to generate  $T_1$ estimations online.

\textbf{(III) Adaptive residual generation with Koopman model: }
The offline-trained KAN model fails to account for unforeseen changes in operating conditions. Koopman module, on the other hand, learns the variability in battery dynamics from $T_2$ measurements, thereby overcoming the poor adaptability of the KAN network.
Thus, the proposed KAN-Koopman diagnostic algorithm leverages the unique advantages of each method to ensure rapid detection of thermal anomalies with improved adaptability to changes in battery dynamics.

\subsection{KAN-based estimation} To estimate the core temperature $T_1$, we adopt the KAN-Therm model proposed in our previous work \cite{mallick2025kan}. The proposed KAN model uses four input features as:
\begin{align} \label{T1 KAN}
    \mathbb{T}_{1_k} = \text{KAN}\left( \boldsymbol{\alpha}_k\right), \,\,\, \boldsymbol{\alpha}_k = \begin{bmatrix}
   T_{2_k} & T_{\infty_k} & I_k & \dot{Q}_{c_k}
\end{bmatrix}^T.
\end{align}
Here, $\mathbb{T}_{1_k}$ is the estimation of  ${T}_1$ at $k^{th}$ instant, and the estimation error   $e_k = \mathbb{T}_{1_k} - T_{1_k}$.
From \eqref{core temp}, we can rewrite the core temperature dynamics as $ T_{1_{k}} = F(\boldsymbol{\alpha}_{k}),$  where $F(\boldsymbol{\alpha}_{k})$ is the solution of \eqref{core temp}. Since the underlying battery dynamics is Lipschitz continuous, we can find a  Lipschitz constant $L_F \in \mathbb{R}^+$, such that $F$ is Lipschitz continuous in $\boldsymbol{\alpha}_k$: 
\begin{align}\label{F lipschtiz}
   \Vert F(\boldsymbol{\alpha}_k)\! - \!F(\boldsymbol{\alpha}_k + \Delta \boldsymbol{\alpha}_k) \Vert \!\leqslant \! L_F\Vert  \Delta \boldsymbol{\alpha}_k\Vert, \,\,\, \forall k.
\end{align}
Under thermal anomalies, the true $T_{1_{k}} = F( \boldsymbol{\alpha}_k +\delta_k) +  \delta_{1_k}$, whereas the KAN network estimates $\mathbb{T}_{1_k}$ using 
      $\boldsymbol{\alpha}_k$. Here, $\delta_k := \begin{bmatrix}
  0 & \!\!\! 0 & \!\! \delta_{2_k}  & \!\! \!0 
\end{bmatrix}^T\!\!\neq 0$ and $\delta_{1_k}\neq 0$ is a function of thermal fault $\Bar{\delta}_{1_k}$.
Thus, we can write $e_k$ in terms of $\delta_{k}$ and $\delta_{1_k}$  as
\vspace{-1mm}
    \begin{align}\label{corrupted KAN error}
    \Vert e_k(\delta_{k},\delta_{1_k})  \Vert   = \Vert   \text{KAN}\left(\boldsymbol{\alpha}_k \right) -F(\boldsymbol{\alpha}_{k} + \delta_{k}) -   \delta_{1_k} \Vert.
\end{align}
Under nominal battery operation $\Vert \delta_{k} \Vert = 0$ and ${\delta}_{1_k} = 0$, the estimation error  $e_k (\textbf{0},0) = \mathbb{T}_{1_k} - T_{1_k} = \text{KAN}\left(\boldsymbol{\alpha}_k\right) - F(\boldsymbol{\alpha}_{k})$. Hence, using the KAN approximation Theorem~\ref{KAN aprox thm}, 
\begin{align} \label{KAN err lim}
   \vert  e_k (\textbf{0},0) \vert_\infty   \leqslant \mathcal{D}G^{-\kappa-1}, \qquad  \forall k.
\end{align}
 
\begin{Lemm}[Lipschitz  continuity of error with anomalies]\label{LMM KAT under delta}
    Consider the bounded multi-variate  function $F( \boldsymbol{\alpha}_k)$  that is Lipschitz continuous in $\boldsymbol{\alpha}_k$ \eqref{F lipschtiz} and can be expressed with  $\text{KAN}\left( \boldsymbol{\alpha}_k \right)$ network \eqref{T1 KAN}. Under nominal battery operation ($\Vert \delta_{k} \Vert,{\delta}_{1_k} \!\!= \!0$), \eqref{KAN err lim} provides the approximation error bound for $\text{KAN}\left( \boldsymbol{\alpha}_k \right)$. 
    Then, under anomalous input $\Vert \delta_{k} \Vert, \delta_{1_k} \!\!\neq \!0$, the KAN approximation error $e_k (\delta_k, \delta_{1_k})$  is Lipschitz  continuous in $\delta_k, \delta_{1_k},$ with  Lipschitz constant $L_e\!>\!0$, i.e., 
     \begin{align}
       \Big|  \Vert e_k(\delta_k,\delta_{1_k})  \Vert \! - \!\Vert e_{k}(0,0) \Vert \Big | \!\leqslant \!L_e \Big[\Vert  \delta_{1_k}  \Vert\! + \! \Vert \delta_{k} \Vert \Big], \,\, \forall k. \label{ek lipschitz}
    \end{align}
\end{Lemm}

\begin{proof}
Under thermal anomaly, consider \eqref{corrupted KAN error} to add and subtract $F(\boldsymbol{\alpha}_{k})$ on the right-hand side (RHS) to obtain
\vspace{-1mm}
\begin{align}
    \Vert e_k \Vert  & = \Vert \text{KAN}\left(\boldsymbol{\alpha}_k \right) -F(\boldsymbol{\alpha}_{k}) -     \delta_{1_k} +  F(\boldsymbol{\alpha}_{k}) - F(\boldsymbol{\alpha}_{k} + \delta_{k}) \Vert. \label{added nom F}
\end{align}
Then,  using the triangle inequality, \eqref{added nom F} becomes 
\vspace{-1mm}
\begin{align}
    \Vert e_k \Vert & \leqslant  \Vert  \text{KAN}\left(\boldsymbol{\alpha}_k \right) -F(\boldsymbol{\alpha}_{k}) \Vert \,\,    + \Vert   \delta_{1_k}  \Vert  + \Vert F(\boldsymbol{\alpha}_{k} )  - F(\boldsymbol{\alpha}_{k} + \delta_{k}) \Vert.
\end{align}
Next, using Lipschitz continuity of $F$ and \eqref{KAN err lim} yields
\begin{align} 
    \Vert e_k \Vert & \leqslant  \Vert  e_k (\textbf{0},0) \Vert \, + \Vert   \delta_{1_k}  \Vert + \, L_F \Vert \delta_{k} \Vert. \label{KAN upper bound}
\end{align}
Defining $L_e := \max(1 , L_F)$, we obtain
\begin{align} 
    \Vert e_k \Vert & \leqslant  \Vert  e_k (\textbf{0},0) \Vert \, + L_e \left(\Vert   \delta_{1_k}  \Vert + \, \Vert \delta_{k} \Vert\right). \label{KAN upper bound1}
\end{align}
subtracting $\Vert  e_k (\textbf{0},0) \Vert$ on each side of the inequality and  we prove \eqref{ek lipschitz}.

\end{proof}

\begin{remm}
    Lemma \ref{LMM KAT under delta} implies that the error $e_k (\delta_k, \delta_{1_k})$ in KAN estimate is Lipschitz continuous in $\delta_k, \delta_{1_k}$, and thus, also  absolutely continuous in $\delta_k, \delta_{1_k}$. This regularity of $e_k (\delta_k, \delta_{1_k})$ in $\delta_k, \delta_{1_k}$ implies that the KAN-based core temperature estimation $\mathbb{T}_1$ will vary regularly with thermal anomalies. Thus, the KAN estimate can be utilized reliably as a measurement proxy even under anomalous thermal conditions for the following Koopman module. This result is strong as it guarantees the KAN network behavior under anomalous inputs that are beyond the nominal training set. 
\end{remm}

\subsection{Koopman-based prediction}
We use KO to generate the core  and surface temperature predictions  $\widehat{\mathbb{T}}_1$ and $\widehat{T}_2$ using KAN estimate and available data as
$\mathbb{y}_k = \begin{bmatrix}
    \mathbb{T}_{1_k} & T_{2_k}
\end{bmatrix}^T$ and $u_k = \begin{bmatrix}
    I_k & T_{\infty_k}
\end{bmatrix}^T$. Then, adopting the delay-embedding technique discussed in Section~\ref{KO_preli} and using \eqref{ko opt}, we can first learn the Koopman linear model over a sliding learning window of length $W_l$ and then generate the predictions $\widehat{\mathbb{y}}$ and the detection residual $\mathbb{r}_k$ over the following prediction window of length $W_p$ as:
\vspace{-1mm}
\begin{align}\label{z KAN}
    \mathbb{z}_{k+1} = \mathbb{A} \mathbb{z}_k + \mathbb{B} u_k; \,\,\, \widehat{\mathbb{y}}_k = \mathbb{C} \mathbb{z}_k, \,\,\, \mathbb{r}_k = \Vert \mathbb{y}_k - \widehat{\mathbb{y}}_k \Vert.
\end{align}

\subsection{Diagnostic guarantees}
\noindent To obtain the diagnostic guarantees for the residual $\mathbb{r}_k$, let us first consider a hypothetical Koopman model that can access the true core temperature measurement $T_1$ such that ${y}_k \!=\! \begin{bmatrix}
    {T}_{1_k} \!& \!T_{2_k}
\end{bmatrix}^T$ and the prediction $\widehat{y}$,  residual $r_k$ are defined as
\vspace{-3mm}
\begin{align}\label{z true T1}
    {z}_{k+1} = {A} {z}_k + {B} u_k, \, \,\, \widehat{{y}}_k = {C} {z}_k, \,\,\,  {r}_k = \Vert {y}_k - \widehat{{y}}_k \Vert.
\end{align}
Our previous work in \cite{ghosh2024koopman1} shows that in the presence of anomalies, for a Koopman model that learns from true system measurement, $r_k$  will cross the threshold reliably to detect the anomalous input. Hence, we assume that for the hypothetical model described above, there exist a constant $M_1 \!\in \!\mathbb{R}^+$ and a small practical relaxation constant $\epsilon_M\!>\!0$, such that 
\begin{align}\label{r m att}
    \Vert y_k - C z_k \Vert= r_k  \geqslant M_1 \left[\Vert \delta_k\Vert + \Vert \delta_{1_k}\Vert\right] - \epsilon_M, \quad \forall k.
\end{align}
This sensitivity constant $M_1$ is a design parameter, and ensures detection of all attacks that satisfy $ \Vert \delta_k\Vert + \Vert \delta_{1_k}\Vert \geqslant \frac{\epsilon_M}{M_1}.$ This bound shows that higher value of sensitivity $M_1$ ensures the detection of smaller-valued $\delta_k,\delta_{1_k}$. 
Furthermore, under no thermal anomalies, the hypothetical residual  must remain below the threshold to avoid false detection, i.e. 
\begin{align}
    r_k  \! \leqslant \!\epsilon_M, \qquad \text{when,}\,\, \Vert \delta_k \Vert,\delta_{1_k}=0. \label{nom r bound}
\end{align}
     \noindent
 Next, we will present the two main theorems for our KAN-Koopman framework using these conditions on $r_k$.

\begin{thmm}[Upper bound on residual under no anomaly] \label{thmm threshold}
      Consider the proposed KAN-Koopman model \eqref{z KAN} learned using the KAN estimate $\mathbb{T}_1$  and the hypothetical Koopman model \eqref{z true T1}  using true measurement $T_1$. Then, during nominal battery operation ($\Vert \delta_{k} \Vert, \delta_{1_k} \!= \!0$),
    the diagnostic residual $\mathbb{r}_k$ is bounded for $\beta_1, \beta_2, \varsigma \in \mathbb{R}^+$ as: 
    \begin{align}\label{nom r lemm}
      \mathbb{r}_k    \leqslant  \epsilon_M  +    \sqrt{2} \,\varsigma\,\mathcal{D} G^{-\kappa-1},\,\,\, \forall k.
  \end{align}
\end{thmm}
\begin{proof}
    Let us expand the residual $\mathbb{r}_k$ as:
\begin{align} 
    \mathbb{r}_k = \Vert \mathbb{y}_k - \mathbb{C} \mathbb{z}_k\Vert& = \Vert (\mathbb{y}_{{k}} - y_k) +( y_k - C z_k )  + (C z_k-C \mathbb{z}_k )+ (C \mathbb{z}_k -  \mathbb{C} \mathbb{z}_k  ) \Vert. \label{r expand}
\end{align}
Defining $\Delta C := \mathbb{C} - C$, we  use the triangle inequality on RHS of \eqref{r expand} and substitute \eqref{z true T1} to obtain
\begin{align}\label{r 4 parts}
    \mathbb{r}_k   & \leqslant  r_k + \Vert \mathbb{y}_{{k}} - y_k \Vert +\Vert C \Vert \,\Vert z_k -  \mathbb{z}_k \Vert + \Vert \Delta C \mathbb{z}_k \Vert.
\end{align}
Let us denote  $\xi_k = \mathbb{y}_{{k}} - y_k $. Then, $\xi_k = \begin{bmatrix}
    e_k & 0
\end{bmatrix}^T$ since the surface temperature $T_2$ is known. Hence, we obtain
\begin{align}
    \Vert \xi_k \Vert = \Vert \mathbb{y}_{{k}} - y_k \Vert = \Vert e_k \Vert. \label{xi as e}
\end{align}
Next, the KMs can be found analytically using \eqref{ko opt} as 
\begin{align}\label{KMs}
    \mathbb{C} = \mathbb{Y}_b \mathbb{X}_b^T \left(\mathbb{X}_b\mathbb{X}_b^T \right)^{-1}, \qquad {C} = {Y}_b {X}_b^T \left({X}_b {X}_b^T \right)^{-1}.
\end{align}
Here, $\mathbb{Y}_b$ and ${Y}_b$ are the delay-embedded measurements  
for the proposed and hypothetical Koopman models, respectively. Additionally, $\mathbb{z}_k$ and $z_k$ are the single columns in the Hankel matrices $\mathbb{X}_b$ and $X_b$, respectively. These matrices $\mathbb{X}_b$ and $X_b$ contain, respectively, a time-delayed sequence of $\mathbb{y}, u$ data and $y,u$ data from previous $k-W_l$ observations. Thus,  $\Delta z_k = \mathbb{z}_k-z_k$, $\Xi = \mathbb{Y}_b - Y_b$, and $\zeta = \mathbb{X}_b - X_b$ contain $0$ and $\xi_j$ as their entries for $j \in  \mathcal{J} \coloneqq \{k, \cdots,k-W_l \}$. Let us now define a constant $\beta_1 >0$ such that $ \beta_1 \Vert \xi_k \Vert \geqslant W_l \max_{j \in \mathcal{J} } (\Vert \xi_j\Vert)$ and then, we can obtain the bounds as 
\begin{align}\label{bound small z}
    & \Vert \Delta z_k \Vert  \leqslant  \beta_1 \Vert \xi_k \Vert, \,\,\, \Vert \Xi \Vert \leqslant  \beta_1 \Vert \xi_k \Vert, \,\,\, \Vert\zeta\Vert \leqslant  \beta_1 \Vert \xi_k \Vert.
\end{align}
Now, substituting $\mathbb{Y}_b = Y_b - \Xi$ and $\mathbb{X}_b = X_b - \zeta$ in the first equation of \eqref{KMs},  we obtain
\begin{align}
    \mathbb{C} &=  [Y_b - \Xi][X_b-\zeta]^T\Big([X_b-\zeta][X_b-\zeta]^T\Big)^{-1}\\
    & = [Y_b - \Xi] [X_b-\zeta]^T \, \Big(X_b X_b^T \!- \eta \Big)^{-1}, 
\end{align}
where $\eta \!= \!X_b\zeta^T \!+ \!\zeta X_b^T\! - \!\zeta\zeta^T$. Now, taking $l_2$ norm and using the triangle inequality, we can write
$
    \Vert \eta \Vert \!\leqslant\! 2\Vert X_b \Vert \Vert \zeta \Vert \!+\! \Vert \zeta \Vert^2 .
$
Next, using the boundness of the Hankel matrix $X_b$ and $\zeta$ from \eqref{bound small z},
we can find a constant $\beta_3 \in \mathbb{R}^+$ such that  we get $\Vert \eta \Vert \leqslant \beta_3 \Vert \xi_k\Vert$.
Next, multiplying  $\left(X_b X_b^T\right)^{-1}X_b X_b^T$ with $(X_b X_b^T - \eta )^{-1}$  from left and merging $X_b X_b^T$, we get
\begin{align}
    \mathbb{C} &=  [Y_b - \Xi][X_b-\zeta]^T\!\left(X_b X_b^T\right)^{-1}\!\Big(I \!- N \Big)^{-1}\nonumber \\
    & = C (I-N)^{-1} - {\mu}, \quad \text{where} \,\, N   = \!\left(X_b X_b^T\right)^{-1}\!\eta,\label{c mid}\\
    \mu  &= \left[Y_b \zeta^T\! +\! \Xi X_b^T \! - \!\Xi \zeta^T\right]\!\left(X_b X_b^T\right)^{-1} (I-N)^{-1}. \label{mu}
\end{align}
Now, taking $l_2$ norm on $\mu$ and considering the boundness of $X_b$ and $Y_b$, we can use \eqref{bound small z}  to obtain a $\beta_4 \in \mathbb{R}^+$ such that $\Vert \mu \Vert \!\leqslant \!\beta_4 \Vert \xi_k \Vert$.
Next, subtracting $C$ on both sides of \eqref{c mid}, using definition of $N$ \eqref{c mid} after simplification,  we can obtain
\begin{align}
    \Delta C & = C  \left(X_b X_b^T\right)^{-1}\eta \left(I - \left(X_b X_b^T\right)^{-1}\eta \right)^{-1} - \mu.
\end{align}
Notice that $\Delta C$ contains $\eta$, $\mu$, and bounded terms. Thus, using the bounds on $\eta$, $\mu$, we can obtain $\beta_2 \in \mathbb{R}^+$ such that
\vspace{-1mm}
\begin{align}
    \Vert \Delta C \Vert \leqslant \beta_2  \Vert \xi_k \Vert . \label{del c bound}
\end{align}
Finally, using \eqref{xi as e}, \eqref{bound small z}, and \eqref{del c bound}, we can rewrite \eqref{r 4 parts}  as:
\vspace{-1mm}
  \begin{align}\label{nom r eVert}
      \mathbb{r}_k   &  \leqslant r_k  +  \left(1 + \varsigma_k   \right) \Vert e_k \Vert,
  \end{align}
 where $0\leqslant\varsigma_k = \beta_1 \Vert C \Vert  + \beta_2 \Vert \mathbb{z}_k \Vert$. Now, since the Koopman prediction $Cz$ and $\mathbb{C}\mathbb{z}$ remain bounded over the prediction window of length $W_p$, we can find a constant $\varsigma \in \mathbb{R}^+$ such that $1\leqslant 1 + \varsigma_k \leqslant \varsigma <\infty,\,\, \forall k$. Hence, in the absence of thermal anomalies, i.e., $\Vert \delta_{k} \Vert, \delta_{1_k} = 0$, \eqref{nom r eVert} becomes
 \vspace{-1mm}
    \begin{align}\label{nom r epsi}
      \mathbb{r}_k   &  \leqslant r_k  +  \varsigma \,\Vert  e_k (\textbf{0},0) \Vert .
  \end{align}
  Additionally, from \eqref{KAN err lim}, and using definitions of  $l_2$ and $l_\infty$ norm, we can get 
  \begin{align} \label{ek l2 norm}
      \Vert  e_k (\textbf{0},0) \Vert \leqslant \sqrt{2} \vert  e_k (\textbf{0},0) \vert_\infty \leqslant \sqrt{2} \mathcal{D}G^{-\kappa-1}.
  \end{align} 
  Finally, substituting this bound for $\Vert  e_k (\textbf{0},0) \Vert$ and $r_k$  \eqref{nom r bound} in \eqref{nom r epsi}, we can obtain \eqref{nom r lemm}. This completes the proof.
\end{proof}
\noindent
Theorem~\ref{thmm threshold} shows that during nominal battery operation ($\Vert \delta_{k} \Vert, \delta_{1_k} =0$), the residual $\mathbb{r}_k$ remains close to zero since $r_k$ is bounded by $\epsilon_M$ and the KAN approximation can be made arbitrarily small depending on the network structure i.e. with higher grid $G$ and activation function order $\kappa$. 

  \begin{thmm}[KAN-Koopman guarantees for detection] \label{thmm detection}
      Consider the proposed KAN-Koopman model \eqref{z KAN} learned using the KAN estimate $\mathbb{T}_1$ and the hypothetical Koopman model \eqref{z true T1} learned using true measurement $T_1$. 
      Then, for thermal anomalies $\delta_{k}$ and $ \delta_{1_k}$ satisfying,
      \vspace{-1mm}
    \begin{align}
        \Vert  \delta_{1_k}  \Vert +  \Vert \delta_{k} \Vert \geq \epsilon_M + \sqrt{2} \, \mathcal{D} G^{-\kappa-1}, \label{delta norm}
    \end{align}
      the residual KAN-Koopman $\mathbb{r}_k$ from \eqref{z KAN} has the following lower bound for positive constant $M\in \mathbb{R}^+$:
      \vspace{-1mm}
    \begin{align} 
      \mathbb{r}_k   &  \geqslant M \Big[\Vert  \delta_{1_k}  \Vert +  \Vert \delta_{k} \Vert\Big] - \epsilon_M, \quad \forall k. \label{att r thmm}
    \end{align}

  \end{thmm}
  \begin{proof}
  Using the definitions provided in Theorem 2, let us first rewrite the terms in \eqref{r expand} as
  \begin{align} 
    \mathbb{r}_k =  \Vert (y_k - C z_k) - ( C \Delta \mathbb{z}_k +\Delta C \mathbb{z}_k - \xi_k)  \Vert. \label{r th2}
\end{align}
To prove this theorem, we will first show that 
 \begin{align} \label{enable_TE}
         \Vert  C \Delta \mathbb{z}_k +\Delta C \mathbb{z}_k - \xi_k   \Vert \leqslant \Vert y_k - C z_k \Vert.
     \end{align}
and then apply the inverse triangle inequality on \eqref{r th2} as
       \begin{align}\label{att r eVert}
      \mathbb{r}_k   &  \geqslant \Vert y_k - C z_k \Vert -  \Vert  C \Delta \mathbb{z}_k +\Delta C \mathbb{z}_k - \xi_k   \Vert.
  \end{align}
To achieve this goal, we will first use the triangle inequality on the second term to write:
\begin{align}
    &\Vert  C \Delta \mathbb{z}_k +\Delta C \mathbb{z}_k - \xi_k   \Vert \leqslant  \Vert  C  \Vert  \, \Vert \Delta \mathbb{z}_k  \Vert  + \Vert \Delta C \Vert \, \Vert \mathbb{z}_k \Vert +  \Vert \xi_k  \Vert.
\end{align}
Next, using the bounds from \eqref{bound small z} and \eqref{del c bound}, the above inequality becomes 
    $\Vert  C \Delta \mathbb{z}_k +\Delta C \mathbb{z}_k - \xi_k   \Vert \leqslant  \varsigma \, \Vert e_k \Vert.$ 
Subsequently, we first use \eqref{KAN upper bound1} from  Lemma~1 and then \eqref{ek l2 norm} from Theorem~2 to this inequality to yield
\begin{align}
    \Vert  C \Delta \mathbb{z}_k +\Delta C \mathbb{z}_k - \xi_k   \Vert & \leqslant  \sqrt{2}\, \varsigma \, \mathcal{D} G^{-\kappa-1} + \varsigma \,  L_e \Big[\Vert  \delta_{1_k}  \Vert +  \Vert \delta_{k} \Vert \Big]. \label{del o}
\end{align}
     Now,  for thermal anomalies $\delta_{k}$ and $ \delta_{1_k}$ satisfying \eqref{delta norm}, 
     \begin{align}\nonumber
           \sqrt{2}\, \varsigma \mathcal{D} G^{-\kappa-1} &\leqslant  \varsigma \Big[\Vert  \delta_{1_k}  \Vert +  \Vert \delta_{k} \Vert \Big]- \epsilon_M \varsigma,\\
            & \leqslant  \varsigma \Big[\Vert  \delta_{1_k}  \Vert +  \Vert \delta_{k} \Vert \Big]-\epsilon_M , \,  \text{(as $\varsigma \geqslant 1$).} 
     \end{align}
      Substituting the above inequality in \eqref{del o} and defining $M_2=\varsigma (1+L_e)$, yields
     \begin{align}
         \Vert  C \Delta \mathbb{z}_k+\Delta C \mathbb{z}_k - \xi_k   \Vert&\leqslant M_2 \Big[\Vert  \delta_{1_k}  \Vert +  \Vert \delta_{k} \Vert \Big]-\epsilon_M,\label{M2 delta} \\
         &\leqslant M_2 \Big[\Vert  \delta_{1_k}  \Vert +  \Vert \delta_{k} \Vert \Big] \label{M2_delta_2},
     \end{align}
     as $\epsilon_M>0$. Now, we can choose the sensitivity parameter in \eqref{r m att} as $M_1 >M_2$ to claim that $\Vert y_k - C z_k \Vert = r_k \geqslant M_2 \left[\Vert \delta_k\Vert + \Vert \delta_{1_k}\Vert\right] - \epsilon_M$. 
   Combining this inequality and \eqref{M2 delta}, we obtain the desired inequality defined in \eqref{enable_TE}.  
Lastly, using \eqref{r m att} and \eqref{M2_delta_2}  in the two terms of the inequality \eqref{att r eVert} respectively, we obtain \eqref{att r thmm} by defining $M=M_1-M_2$. This completes the proof.
  \end{proof}

  \noindent
  Theorem \ref{thmm detection} implies that if the sensitivity parameter $M_1$ is chosen as $M_1 \geqslant M_2$,
  the diagnostic residual $\mathbb{r}_k$ can ensure the reliable detection of thermal anomalies $\delta_{1_k}$ and $\delta_k$ that satisfy \eqref{delta norm}, where \eqref{delta norm} indicates sensitivity of the method to thermal anomalies with magnitude higher than the bound of the KAN approximation error $ e_k(\textbf{0},0)$ and the practical relaxation factor $\epsilon_M$. Thus, if the bound for  $ e_k(\textbf{0},0)$, i.e., $\sqrt{2} \mathcal{D} G^{-\kappa-1}$ is kept small enough with the appropriate network structure (grid $G$ and spline order $\kappa$), the proposed KAN-Koopman algorithm can reliably detect low-magnitude thermal anomalies in battery. 

\section{Simulation Results}\label{sim_result}
 We consider a $2.3Ah$  cylindrical $LiFePO_4-LiC_6$ battery cell for our case studies and adopt the model parameters from \cite{vyas2024input}. 
 We consider a sampling frequency of $100 Hz$ and a zero-mean Gaussian measurement noise of $\pm 0.05^\circ$ for the $T_1, T_2,$ and $T_\infty$ measurements.
Next, we provide details of the KAN-Koopman model and present two thermal anomaly scenarios to evaluate our proposed algorithm's performance. 
\subsection{Details of KAN-Koopman implementation}

\textbf{KAN architecture:} We consider a two-layered KAN network with a network configuration of $[\![4, 3, 1 ]\!]$: input layer with width 4 (four features $T_2, T_{\infty}, I,\dot{Q}_c $), followed by a KAN layer of width 3, and the final KAN layer of width 1 (single output  $\mathbb{T}_1$). 
Furthermore, we consider the grid $G=5$ and spline order $\kappa =3$ for the KAN model. 

\textbf{Koopman parameters:} We consider a learning window of length $3000$ data-points ($30s$) with delay-embedding length of $2100$ to generate predictions for the next $500$ observations ($5s$).
Furthermore, we choose a threshold of $0.03$ based on the nominal fluctuations in residual $\mathbb{r}$ under $\Vert \delta_{k} \Vert, \delta_{1_k} \!= \!0$. We have used the average of generated residuals over a moving window of 3000 to reduce the setting and resetting of the attack detection flag during threshold crossings.

\subsection{Detection of incipient thermal fault}
In this scenario, we consider that the battery is charging at $1C$ rate and an incipient thermal fault $\delta_1$ is injected from 800s. Due to the slow-progressing nature of the injected fault, $T_1$ and $T_2$ do not exhibit any abrupt growth in temperatures, as shown in the top plot of Fig~\ref{fig:Qfault}. 
The second plot of Fig~\ref{fig:Qfault}  shows the charging current, and the third plot illustrates the increased heat generation in the battery after the incipient fault injection. Under such thermal anomaly, our diagnostic residual $\mathbb{r}$ crosses the threshold at 1112s as shown in the fourth plot of Fig~\ref{fig:Qfault}, resulting in a reliable detection of the thermal anomaly after 362s of the fault occurrence (shown by setting of the anomaly flag in the last plot). We also compare this performance  with a baseline Koopman-only algorithm that solely relies on $T_2$ measurement for diagnosis. The last plot of Fig~\ref{fig:Qfault} shows that the baseline method takes 112s longer than the proposed method to generate the anomaly flag. Thus, the proposed method exhibits $24\%$ faster detection of  thermal anomalies in this scenario. 
\begin{figure}[h!]
    \centering
    \includegraphics[width=0.5\linewidth]{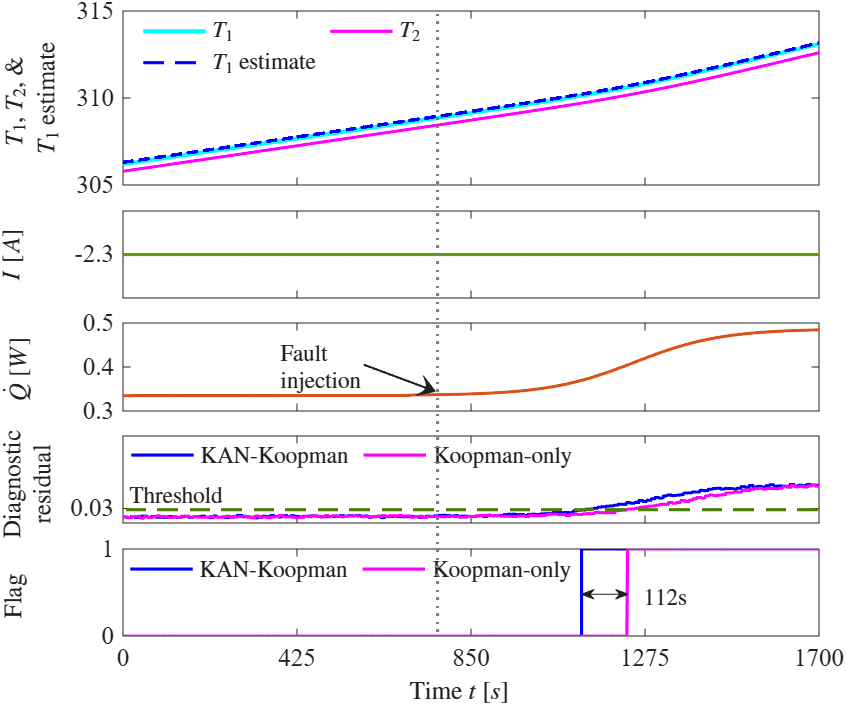}
    \caption{Under an incipient thermal fault, top plot shows $T_1$, $T_2$, and $\mathbb{T}_1$; second plot shows  $I$; third plot shows  $\dot{Q}$;  fourth and fifth plots, respectively, shows the diagnostic residuals  and the thermal anomaly flag generated by proposed and the baseline Koopman algorithm.}
    \label{fig:Qfault}
\end{figure}

\section{Conclusion} \label{conclusion}
In this paper, we propose a model-free KAN-Koopman diagnostic algorithm for the detection of battery thermal anomalies that can be induced from both physical faults and cyberattacks. The proposed method adopts a lightweight KAN network that estimates the core temperature, and next, a Koopman model is learned online using the KAN estimation along with the limited available data to generate  predictions for the battery core and surface temperatures. This online learning approach leads to improved adaptability with changes in battery dynamics, while the KAN–Koopman integrated structure ensures rapid detection of battery thermal anomalies. The diagnostic residual is generated as the error in the Koopman predictions for the KAN-based  core temperature estimation and surface temperature measurement.  Furthermore, we derive analytical conditions for diagnostic guarantees on our KAN-Koopman detection scheme. We compare the performance of our model with baseline Koopman-only model, that only utilizes surface temperature for diagnostic. Our comparison result shows that the proposed KAN-Koopman method exhibits $24\%$ faster detection of an incipient thermal fault induced thermal anomalies compared to the Koopman-only model.


\bibstyle{arxiv}
\bibliography{ref1.bib}

\end{document}